\RequirePackage{amsmath}
\documentclass[runningheads,envcountsame]{llncs}
\usepackage{multicol}
\usepackage{amsfonts}
\usepackage{amssymb}
\usepackage{verbatim}
\usepackage{graphicx}
\usepackage{epic}
\usepackage{eepic}
\usepackage{epsfig,float}
\usepackage{pdfsync}
\usepackage{gastex}
\usepackage{url}

\DeclareMathOperator{\col}{col}
\DeclareMathOperator{\row}{row}
\DeclareMathOperator{\cols}{cols}

\usepackage{color}

\definecolor{aocolour}{rgb}{0.7,0.8,1}
\definecolor{gjcolour}{rgb}{1,0.8,0.7}

\usepackage{multicol}
\renewcommand{\le}{\leqslant}
\renewcommand{\ge}{\geqslant}

\newcommand{\emp}{\emptyset}

\newcommand{\Sig}{\Sigma}

\newcommand{\noin}{\noindent}
\newcommand{\bi}{\begin{itemize}}
\newcommand{\ei}{\end{itemize}}
\newcommand{\be}{\begin{enumerate}}
\newcommand{\ee}{\end{enumerate}}
\newcommand{\bd}{\begin{description}}
\newcommand{\ed}{\end{description}}
\newcommand{\bq}{\begin{quote}}
\newcommand{\eq}{\end{quote}}

\def\shu{\mathbin{\mathchoice
{\rule{.3pt}{1ex}\rule{.3em}{.3pt}\rule{.3pt}{1ex}
\rule{.3em}{.3pt}\rule{.3pt}{1ex}}
{\rule{.3pt}{1ex}\rule{.3em}{.3pt}\rule{.3pt}{1ex}
\rule{.3em}{.3pt}\rule{.3pt}{1ex}}
{\rule{.2pt}{.7ex}\rule{.2em}{.2pt}\rule{.2pt}{.7ex}
\rule{.2em}{.2pt}\rule{.2pt}{.7ex}}
{\rule{.3pt}{1ex}\rule{.3em}{.3pt}\rule{.3pt}{1ex}
\rule{.3em}{.3pt}\rule{.3pt}{1ex}}\mkern2mu}}

\newcommand{\cA}{{\mathcal A}}

\newcommand{\cD}{{\mathcal D}}

\newcommand{\cK}{{\mathcal K}}

\newcommand{\cL}{{\mathcal L}}

\newcommand{\cN}{{\mathcal N}}

\newcommand{\cT}{{\mathcal T}}

\newcommand{\one}{{\mathbf 1}}

\title{On the State Complexity of the Shuffle\\of Regular Languages\thanks{This work was supported by the Natural Sciences and Engineering Research Council of Canada under grant No.~OGP0000871, by VEGA grant 2/0084/15, and by the National Science Centre, Poland under project number 2014/15/B/ST6/00615.}
}
\author{Janusz~Brzozowski \inst{1}, 
 Galina Jir\'askov\'a \inst{2},
Bo Liu \inst{1}\thanks{Present address: 
Google Inc.,
1600 Amphitheatre Parkway,
Mountain View, CA 94043,
USA.}, 
Aayush Rajasekaran\inst{1}, \and Marek Szyku{\l}a \inst{3}
}

\authorrunning{Brzozowski, Jir\'askov\'a, Liu, Rajasekaran, Szyku{\l}a}

\institute{David R. Cheriton School of Computer Science, University of Waterloo, \\
Waterloo, ON, Canada N2L 3G1\\
\{{\tt \{brzozo,  b23liu, arajasek\}@uwaterloo.ca}\}
\and
Mathematical Institute,
Slovak Academy of Sciences,\\
Gre\v s\'akova 6, 040 01 Ko\v sice, Slovakia \\
\email{jiraskov@saske.sk}
\and
Institute of Computer Science, University of Wroc{\l}aw,\\
Joliot-Curie 15, PL-50-383 Wroc{\l}aw, Poland\\
\{{\tt msz@cs.uni.wroc.pl}\}
}

\begin{document}
\maketitle
\begin{abstract}
We investigate the shuffle operation
on regular languages represented 
by complete deterministic finite automata.
We prove that 
$f(m,n)=2^{mn-1} + 2^{(m-1)(n-1)}(2^{m-1}-1)(2^{n-1}-1)$
is an upper bound on the state complexity
of the shuffle of two regular languages having
state complexities $m$ and $n$, respectively.
We also state partial results about the tightness of this bound.
We show that there exist witness languages meeting the bound 
if $2\le m\le 5$ and   $n\ge2$, and also if $m=n=6$.
Moreover, we prove that in the subset automaton of the NFA accepting the shuffle, all $2^{mn}$ states can be distinguishable, 
and an alphabet of size three suffices for that.
It follows that the bound can be met if all $f(m,n)$ states are reachable.
We know that an alphabet of size at least $mn$ 
is required provided that $m,n \ge 2$.
The question of reachability, and hence also of the tightness of the bound $f(m,n)$ in general, remains open.
\medskip

\noin
\end{abstract}

\section{An Upper Bound for the Shuffle Operation}
\label{sec:bound}

The \emph{state complexity of a regular language} $L$~\cite{Yu01} is the number of states in a complete minimal  deterministic finite automaton (DFA) recognizing the language; it
will be denoted by $\kappa(L)$.
The \emph{state complexity of an operation} on regular languages is the maximal state complexity of the result of the operation expressed as a function of the state complexities of the operands.

Let $\Sig$ be a finite non-empty alphabet. The \emph{shuffle} $u\shu v$ of words $u,v\in\Sig^*$ is defined as follows:
$$ u\shu v= \{u_1v_1\cdots u_kv_k \mid u=u_1\cdots u_k,  v=v_1\cdots v_k,
u_1,\ldots,u_k,v_1,\ldots, v_k\in \Sig^*\}.$$
The shuffle of two languages $K$ and $L$ over $\Sig$ is defined by
$$ K\shu L=\bigcup_{u \in K, v \in L} u\shu v.$$
Note that the shuffle operation is commutative on both words and languages.

The state complexity of the shuffle operation was first studied by C\^ampeanu, Salomaa, and Yu~\cite{CSY02}, but 
they considered only bounds for incomplete deterministic automata.
 In particular, they proved that $2^{mn}-1$ is a tight upper bound for that case.
Since we can convert an incomplete deterministic automaton into complete one by adding the empty state, it follows that $2^{(m-1)(n-1)}-1$ is a lower bound for the case of complete deterministic automata.
Here we show that this lower bound can be improved, and we derive an upper bound
for two regular languages
represented by complete deterministic automata,
but the question whether this bound is tight remains open.
 
A \emph{nondeterministic finite automaton} (NFA)
is a quintuple $\cA=(Q,\Sigma,\delta,s,F)$,
where $Q$ is a finite non-empty set of states,
$\Sigma$ is a finite alphabet of input symbols,
$\delta\colon Q\times\Sigma\to 2^Q$ is the transition function
which is extended to the domain $2^Q\times\Sigma^*$
in the natural way,
$s\in Q$ is the  initial state, and
$F\subseteq Q$ is the set of final states.
The language accepted by NFA $\cA$ is the set of words
$L(\cA)=\{w\in\Sigma^*\mid \delta(s,w)\cap F \neq \emptyset\}$.

An NFA $\cA$ is \emph{deterministic and complete} (DFA)
if $|\delta(q,a)|=1$
for each $q$ in $Q$ and each $a$ in $\Sigma$.
In such a case, we write  $\delta(q,a)=q'$ 
instead of $\delta(q,a)=\{q'\}$.
A DFA is \emph{minimal} (with respect to the number of states) if all its states are reachable,
and no two distinct states are equivalent.

Every NFA $\cA=(Q,\Sigma,\delta,s,F)$ can be converted
to an equivalent DFA $\cA'=(2^Q,\Sigma,\delta, \{s\},F')$,
where $F'=\{R\in2^Q \mid R\cap F\neq\emptyset\}$.
The DFA $\cA'$ is called
the \emph{subset automaton} of NFA $\cA$. 
The subset automaton may not be minimal since some of its states 
may be unreachable or equivalent  to other states.

Let $K$ and $L$ be regular languages
over an alphabet $\Sigma$
recognized by deterministic finite automata
$ \cK   =(Q_K,\Sig,\delta_K,q_K,F_K)$ 
and $ \cL =(Q_L,\Sig,\delta_L,q_L,F_L)$, respectively.
Then $K\shu L$ is accepted by the  nondeterministic finite automaton
$$\cN=(Q_K \times Q_L, \Sigma, \delta, (q_K,q_L), F_K\times F_L),$$  
where 
$$
     \delta((p,q),a)=\{(\delta_K(p,a),q), (p,\delta_L(q,a))\}.
$$
\label{nfaN}

Let $\cD=(2^{Q_K\times Q_L},\Sigma,\delta',\{(q_K,q_L)\},F')$ 
be the subset automaton of $\cN$.
If $|Q_K|=m$ and $|Q_L|=n$,
then NFA $\cN$
has $mn$ states.
It follows that 
DFA $\cD$ has at most $2^{mn}$
reachable and pairwise distinguishable states.
However, this upper bound cannot be met, as we will show.

In the sequel, we assume that $Q_K=\{1,2,\ldots,m\}$, $q_K=1$, 
$Q_L=\{1,2,\ldots,n\}$, and $q_L=1$.
We say that a state $(p,q)$ of NFA $\cN$
is in row $i$ if $p=i$,
and it is in column $j$ if $q=j$.

\begin{proposition}\label{pq}
 Let $a\in \Sigma$.
 Let $S$ be a state of $\cD$.
 Let $\pi_{\col}(S) = \{p \mid (p,q) \in S\text{ for some }q\}$, and
 $\pi_{\row}(S) = \{p \mid (p,q) \in S\text{ for some }p\}$.
 Then $\pi_x(S)\subseteq\pi_x(S\cdot a)$ for $x \in \{\col,\row\}$.
\end{proposition}
\begin{proof}
 Let $p \in \pi_{\col}(S)$; then we have $(p,q) \in S$ for some $q$.
 Since $\delta((p,q),a)=\{(\delta_K(p,a),q), (p,\delta_L(q,a)\}$,
 we have
 $(p,\delta_L(q,a))\in \delta(S,a)$, so $p \in \pi_{\col}(\delta(S,a))$.
 By symmetry, the same claim holds for $\pi_{\row}$.
\qed
\end{proof}

We claim that 
in the subset automaton $\cD$,
every reachable subset $S$ of \mbox{$Q_K\times Q_L$} 
must contain a state in column 1 and a state in row 1,
that is, it
must satisfy
the following condition.
\medskip

\noin
{\bf Condition  (C):}
There exist states $(s, 1)$ and $(1,t)$ in $S$ for some $s\in Q_K$
and $t\in Q_L$.

\begin{lemma}\label{reach}
  Every reachable subset $S$ of subset automaton $\cD$
  satisfies Condition~(C).
\end{lemma}
\begin{proof}
   The initial subset of $\cD$ is $\{(1,1)\}$,
   and it satisfies Condition~(C).
By Proposition~\ref{pq}, for every $a \in \Sigma$ we get that $1 \in \pi_{\col}(\delta(S,a))$ and $1 \in \pi_{\row}(\delta(S,a))$, so $\delta(S,a)$ satisfies Condition~(C).
 By induction, all reachable subsets
 satisfy Condition~(C).
\qed
\end{proof}

\begin{theorem}[Shuffle: Upper Bound]
\label{prop:bound}
 Let $\kappa(K)=m$ and $\kappa(L)=n$.
 Then 
 the state complexity of the shuffle of  $K $ and $L $ is at most
\begin{equation}
\label{eq:bound}
f(m,n)=2^{mn-1} + 2^{(m-1)(n-1)}(2^{m-1}-1)(2^{n-1}-1).
\end{equation}
\end{theorem}
\begin{proof}
 By Lemma~\ref{reach},
 every reachable subset of $\cD$ must contain
 a state in row~1 and a state in column 1.
 There are $2^{mn-1}$ subsets containing state 
 $(1,1)$,
 and $2^{(m-1)(n-1)}(2^{m-1}-1)(2^{n-1}-1)$ subsets
 not containing $(1,1)$ but containing $(s,1)$ 
 for some $s\in \{2,3,\ldots,m\}$ and $(1, t)$ 
 for some $t\in \{2,3,\ldots,n\}$. 
 This gives $f(m,n)$.
\qed
\end{proof}

Let $K$ and $L$ be two regular languages over $\Sig$.
If $\kappa(K)=\kappa(L)=1$,
then  each of $K$, $L$, and $K \shu L $ is  either $\emp$ or $\Sig^*$,
and $\kappa(K \shu L) =1$;
hence the bound $f(1,1)=1$ is tight.

Now suppose that $\kappa(K)=1$; 
here we have two possible choices for $K$, 
the empty language or $\Sig^*$. 
The first choice leads to $\kappa(K\shu L)=1$. 
Hence  only
the second choice is of interest, where the language
$K \shu L =\Sig^*\shu L $ is the all-sided ideal~\cite{BJL13} 
generated by $L$.
If $\kappa(L)=2$, the upper bound $f(1,2)=2$
is met by the unary language $L=aa^*$.
Hence assume that $\kappa(K)=1$ and $\kappa(L)\ge 3$.
The next  observation 
shows that in such a case,
the tight bound is less than $f(1,n)=2^{n-1}$.

\begin{proposition}[Okhotin~\cite{Okh10}]\label{prop:m=1}
 If $\kappa(L)\ge 3$,
 then the state complexity of $\Sig^*\shu L$ is at most $2^{n-2}+1$, 
 and this bound can be reached only if $|\Sig| \ge n-2$.
\end{proposition}

Okhotin showed that
the language
$L=(a_1\Sig^*a_1 \cup \cdots \cup a_{n-2}\Sig^*a_{n-2}) \Sig^*$, 
where 
$\Sig=\{a_1,\ldots, a_{n-2}\}$, 
meets this bound~\cite{Okh10}. 
This takes care of the case $\kappa(K)=1$ 
and,  by symmetry, of the case $\kappa(L)=1$.
 
In what follows we assume that $m\ge2$ and $n\ge2$.
First, let us show
that the upper bound $f(m,n)$
cannot be met by regular languages defined over a fixed alphabet.

\begin{proposition}\label{prop:alphabet_mn-1}
 Let $K$ and $L$ be regular languages over $\Sigma$ with
 $\kappa(K)=m$ and $\kappa(L)=n$, where $m,n\ge2$.
 If $\kappa(K\shu L)=f(m,n)$,  then $|\Sig| \ge mn - 1$.
\end{proposition}

\begin{proof}
 For $s=2,3,\ldots,m$ and $t=2,3,\ldots,n$ denote
\begin{align*}
 A_s&=\{(1,1),(s,1)\},\\
 B_t &=\{(1,1),(1,t)\}, \\ 
 C_{st} &=\{(s,1),(1,t)\}. 
 \end{align*}
If all the subsets satisfying Condition~(C)
are reachable,
then, in particular, all the subsets $A_s, B_t$, and $C_{st}$
 must be reachable.
 Let us show that all these subsets 
 must be reached from some subsets containing state $(1,1)$
 by distinct symbols.

 Suppose that a set $A_s$ is reached from a reachable set $S$ with $S\neq A_s$
 by a symbol~$a$, that is, we have $A_s=\delta(S,a)$ and $S\neq A_s$.
 The set $A_s$  contains only states in column 1 and rows 1 or $s$.
 By Proposition~\ref{pq},
  the set $S$ may only contain states in column 1 and in rows 1 or $s$,
 that is, we have $S\subseteq\{(1,1), (s,1)\}$.
 Since  $S\neq A_s$, we must  have $S=\{(1,1)\}$.

 By symmetry, each $B_t$ can only be reached from $\{(1,1)\}$.

 Suppose that a set $C_{st}$ is reached 
 from a reachable set $S$ with $S\neq C_{st}$ 
 by a symbol~$a$.
 By Proposition~\ref{pq},
 we must have $S\subseteq\{(1,1),(s,1),(1,t),(s,t)\}$.
 Let us show that $(1,1)\in S$.
 Suppose for a contradiction that $(1,1)\notin S$.
 Then, since $S$ is reachable,
 it must contain a state in column 1 and a state in row 1,
 that is,
 we must have $\{(s,1),(1,t)\}\subseteq S$.
 But then $(s,t)\in S$ since $S\neq C_{st}$.
 However, then  
   $\delta_K(s,a)=1$ and $\delta_L(t,a)=1$
 which  implies that $(1,1)\in\delta((s,1),a)$, 
 and so  $(1,1)\in C_{st}$.
 This is a contradiction.
 Therefore $C_{st}$ is reached from a set containing $(1,1)$.

 Thus each $A_s$ is reached from $\{(1,1)\}$ by a symbol $a_s$,
 each $B_t$ is reached from $\{(1,1)\}$ by a symbol $b_t$,
 each $C_{st}$ is reached from a set containing $(1,1)$ by a symbol $c_{st}$,
 and we must have
 \begin{align*}
   \delta_K(1,a_s)=s  &\text{ and } \delta_L(1,a_s)=1,\\ 
   \delta_K(1,b_t)=1  &\text{ and } \delta_L(1,b_t)=t,\\
   \delta_K(1,c_{st})=s   &\text{ and }  \delta_L(1,c_{st})=t. 
 \end{align*}
 It follows that all the symbols $a_s, b_t$, and $c_{st}$
 must be pairwise distinct.
 Therefore we  have $|\Sigma|\ge m-1 + n-1 + (m-1)(n-1)=mn-1$. 
\qed
\end{proof}

Unfortunately, this lower bound on the size of the alphabet is not tight, as  is demonstrated by the following example:

\begin{example}
 If $t$ is a transformation of the set 
 $ \{1,2,\ldots, n\}$ and $q\in \{1,2,\ldots,n\}$, 
 let $qt$ be the image of $q$ under~$t$.
 Transformation $t$ can now be denoted by $[1t, 2t,\dots, nt]$.
 
\medskip (1)
 If $m=n=2$, we have $f(2,2)=10$. 
 Let $\Sig=\{a,b,c,d\}$, 
 and let the DFAs  $\cK$ and $\cL$  
 be as shown in Fig.~\ref{fig:Wit22}, 
 and let $K$ and $L$ be their languages.
 Then $\kappa(K \shu L)=10$.
 We have used GAP~\cite{GAP4} to show
 that the bound cannot be reached with a smaller alphabet,
 and that the DFAs of Fig.~\ref{fig:Wit22}
 are unique up to isomorphism. 

\begin{figure}[t]
\unitlength 8pt
\begin{center}\begin{picture}(25,5)(0,2)
\gasset{Nh=2.4,Nw=2.4,Nmr=1.2,ELdist=0.4,loopdiam=1.4}
\node(1')(2,4){$1'$}\imark(1')
\node(2')(8,4){$2'$}\rmark(2')
\drawedge[curvedepth=1.4,ELdist=0.5](1',2'){$a,b,c$}
\drawedge[curvedepth=1.4,ELdist=.5](2',1'){$b,c,d$}
\drawloop(1'){$d$}
\drawloop(2'){$a$}
\node(1)(16,4){$1$}\imark(1)
\node(2)(22,4){$2$}\rmark(2)
\drawedge[curvedepth=1.4,ELdist=0.5](1,2){$a,c,d$}
\drawedge[curvedepth=1.4,ELdist=0.5](2,1){$a,b,c,d$}
\drawloop(1){$b$}
\end{picture}\end{center}
\caption{ Witness DFAs  $\cK$ and $ \cL$  for shuffle with $|Q_K|=2$, $|Q_L|=2$.}
\label{fig:Wit22}
\end{figure}

\medskip (2)
 For $m=2$ and $n=3$,
 the minimal size of the alphabet of a witness pair is 6.
 We have verified this by a dedicated algorithm
 enumerating all pairs of non-isomorphic DFAs with $2$ and $3$ states.
 In contrast to the previous case,
 over a minimal alphabet there are more than $60$
 non-isomorphic DFAs of  $L$ --
 even if we do not distinguish them by sets of final states --
 that meet the bound with some $K$.
 One of the witness pairs is described below.

 Let $\Sig = \{a,b,c,d,e,f\}$.
 Let $\cK=(\{1,2\},\Sig,\delta_{K},1,\{2\})$,
 and let $a=[1,2]$, $b=c=[2,1]$, $d=[1,1]$, $e=[2,2]$, and $f=[2,1]$.
 Let $ \cL=(\{1,2,3\},\Sig,\delta_{L},1,\{1\})$,
 and let $a=[2,2,3]$, $b=[2,1,3]$, $c=[1,1,1]$, $d=e=[3,1,2]$, $f=[3,1,1]$.
 Then $\kappa(K \shu L)=44=f(2,3)$.
\end{example}

The bound $mn-1$ on the size of the alphabet is not tight for $m=n=2$, where an alphabet of size four is required.
For any $m,n \ge 2$ the subsets of $\{1,2\} \times \{1,2\}$
satisfying~(C)
must be also reachable, and to reach them we can use only transformations mapping $1$ to either $1$ or $2$. There are only three such transformations counted in Proposition~\ref{prop:alphabet_mn-1}; thus we need one more letter.

\section{Partial Results about Tightness}

To prove that the upper bound $f(m,n)$ 
of Equation~(\ref{eq:bound}) is tight, 
we must exhibit two languages $K$ and $L$ 
with  state complexities $m$ and $n$, respectively, 
such that $\kappa(K \shu L) = f(m,n)$.
As usual, we use DFAs to represent the languages: 
Let $\cK$ and $\cL$ be minimal \emph{complete} DFAs for $K$ and $L$.
We first construct the NFA $\cN$ 
as defined in Section~\ref{sec:bound}, 
and  we consider the subset automaton $\cD$ of NFA $\cN$.
We must then show 
that $\cD$ has $f(m,n)$ states reachable from the initial state $\{(1,1)\}$, 
and that these states are pairwise distinguishable.
We were unable to prove this for all $m$ and $n$, 
but we have some partial results about reachability in Subsection~\ref{sec:reach}, 
and we deal with distinguishability in Subsection~\ref{sec:distinguish}.

\subsection{Reachability}\label{sec:reach}

We performed computations 
verifying reachability of the upper bound 
for small values of $m$ and $n$. 
These results are summarized in Table~\ref{tab:reachability}.

The computation in the hardest case with $m=n=6$ took about 48 days on a computer with AMD Opteron(tm) Processor 6380 (2500 MHz) and 64 GB of RAM.
Moreover, we verified that in all these cases,
every subset of size at least 3 is directly reachable from some smaller subset.
We also verified that for reachability in case of $m=n=3$ an alphabet of size 12 is sufficient, and in case of $m=n=4$ an alphabet of size 50 is sufficient.
Using these results, we are going to prove reachability for all $m,n$
with $2\le m \le 5$ and $n\ge2$.

\renewcommand{\arraystretch}{1.1}
\begin{table}[t]
\caption{Computational verification of reachability of the bound. The fields with $\checkmark^*$ follow from the proofs of Subsection~\ref{sec:reach}.}\label{tab:reachability}
\begin{center}
$\begin{array}{|l|c|c|c|c|c|c|c|}\hline
m \backslash n&\ 2\      &\ 3\      &\ 4\      &\ 5\      &\ 6\      &\ 7\      &\ \ge 8\    \\ \hline
2             &\checkmark&\checkmark&\checkmark&\checkmark&\checkmark&\checkmark&\checkmark^*\\ \hline
3             &          &\checkmark&\checkmark&\checkmark&\checkmark&\checkmark&\checkmark^*\\ \hline
4             &          &          &\checkmark&\checkmark&\checkmark&\checkmark&\checkmark^*\\ \hline
5             &          &          &          &\checkmark&\checkmark&\checkmark&\checkmark^*\\ \hline
6             &          &          &          &          &\checkmark& ?        & ?          \\ \hline
7             &          &          &          &          &          & ?        & ?          \\ \hline
\ge 8         &          &          &          &          &          &          & ?          \\ \hline
\end{array}$
\end{center}
\vskip-20pt
\end{table}

Without loss of generality, 
the set of states of any $n$-state DFA is denoted by $Q_n=\{1,2,\dots,n\}$.
Let $\cT_n$
be the monoid  of all transformations of the set $Q_n$.
Let $p,q\in Q_n$  and $P\subseteq Q_n$.
Let $\mathbf{1}$ denote the identity transformation.
Let $(p\to q)$ denote the transformation 
that maps state $p$ to state $q$ 
and acts as the identity on all the other states.
Let $(p,q)$ denote the transformation that transposes $p$ and $q$.

Here we deal only with reachability, so final states do not matter.
We assume that the sets of final states are empty in this subsection.

Let $\Sigma_{m,n}=\{ a_{s,t} \mid s\in \cT_m \text{ and } t\in \cT_n\}$
be an alphabet consisting of $m^m n^n$ symbols.
If an input $a$ induces transformations 
$s$ in  $\cT_m$ and $t$ in $\cT_n$, 
this will be indicated by $a\colon s;t$.

Define DFAs $\cK_{m,n}=(Q_m,\Sig_{m,n},\delta_m,1, \emptyset)$ and 
$\cL_{m,n}  = (Q_n, \Sig_{m,n}, \delta_n, 1, \emptyset)$,
where
$\delta_m(p,a_{s,t}) = p s$ if $p\in Q_m$ and
$\delta_n(q,a_{s,t}) =q t$ if $q \in Q_n$.
Let $\cN_{m,n}$ be the NFA for the shuffle of languages
recognized by DFAs $\cK_{m,n}$ and $\cL_{m,n}$
as described in Section~\ref{sec:bound},
and let $\cD_{m,n}$ be the subset automaton of $\cN_{m,n}$.
The NFA $\cN_{m,n}$
has alphabet $\Sig_{m,n}$, 
and so has an input letter 
for every pair of transformations in $\cT_m\times \cT_n$.
Therefore the addition of another input letter 
to the DFAs $\cK_{m,n}$ and $\cL_{m,n}$ 
cannot add any new set of states of $\cN_{m,n}$ 
that would be reachable from $\{(1,1)\}$ in $\cD_{m,n}$.

Let $m'\le m$ and $n'\le n$.
Then DFA $\cK_{m',n'}=(Q_{m'},\Sig_{m',n'}, \delta_{m'}, 1, \emptyset)$ (respectively, the DFA
$\cL_{m',n'}=(Q_{n'},\Sig_{m',n'}, \delta_{n'}, 1, \emptyset)$)
is a sub-DFA of $\cK_{m,n}$ (respectively, of $\cL_{m,n}$),
in the sense that $Q_{m'} \subseteq Q_m$, $\Sig_{m',n'} \subseteq \Sig_{m,n}$, and $\delta_{m'} \subseteq \delta_m$.
As well, NFA $\cN_{m',n'}$ is a sub-NFA of $\cN_{m,n}$.
Note that $\cD_{m,n}$ is extremal for the shuffle: every language $K \shu L$, where $K$ and $L$ are languages with state complexities $m$ and $n$ respectively, is recognized by some sub-DFA of $\cD(m,n)$ after possibly renaming some letters.

For the next lemma it is convenient to consider 
a subset $S$ of states $(p,q)$ of $\cN_{m,n}$ 
as an $m\times n$ matrix, where the entry in row $p$ and column $q$
is $(p,q)$ if $(p,q)\in S$, 
and it is empty otherwise.
We first introduce the following notions.

\begin{definition}
 Let $i,i'\in Q_m$, $i\neq i'$, and $j,j'\in Q_n$, $j\neq j'$.\\
(a) 
A row $i'$ \emph{contains} row $i$, 
if $(i,j)\in S$ implies $(i',j) \in S$ for all $j \in Q_n$.\\
(b) 
A column $j'$ \emph{contains} column $j$
 if $(i,j) \in S$ implies $(i,j') \in S$ for all $i \in Q_m$. \\
(c)
A subset of $Q_m \times Q_n$ is \emph{valid}
if it satisfies Condition~(C) from Lemma~\ref{reach}, that is, if it contains a state in row 1 and a state in column~1.
\end{definition}

\begin{lemma}\label{lem:induction}
 Let $S$ be a valid subset of $Q_m \times Q_n$ 
 with the property that there are distinct $i,i'$ or $j, j'$ 
 such that either 
 row $i'$ contains row $i$ or column $j'$ contains column $j$.
 Assume that every
 valid subset $S'$ of $Q_{m'} \times Q_{n'}$, 
 where $m' < m$, or $n' < n$, or $|S'| < |S|$, 
 is reachable in DFA $\cD_{m',n'}$.
 Then $S$ is reachable in  $\cD_{m,n}$.
\end{lemma}

\begin{proof}
If $S$ contains an empty row or column, 
then without loss of generality 
we can renumber the $n$ states of $\cL_{m,n}$
in such a way that column $n$ is the empty column in $S$. 
By the inductive assumption we know 
that $S$ is reachable in $\cD_{m,n-1}$ by some word $w$. 
Since $\cN_{m,n-1}$ is a sub-NFA of $\cN_{m,n}$, 
$S$ is reachable in 
$\cD_{m,n}$ as well by the same word.
Suppose that $S$ has neither an empty row nor an empty column.
By symmetry, it is sufficient 
to consider the case with distinct $i$ and $i'$ 
such that row~$i'$ contains row $i$.
Let $S' = S \setminus \{(i',j) \mid (i,j) \in S\text{ for }j \in \{1,\ldots,n\}\}$.
Since $|S'| < |S|$, the set $S'$  is reachable by assumption.
To obtain $S$, 
we apply the letter that induces 
the transformation $i \to i'; \one$.
\qed
\end{proof}

\begin{lemma}\label{lem:single_element}
 Let $S$ be a valid subset of $Q_m \times Q_n$
 such that there is a column or a row with exactly one element.
 Assume that every
 valid subset $S'$ of $Q_{m'} \times Q_{n'}$,
 where $m' < m$, or $n' < n$, or $|S'| < |S|$, 
 is reachable in $\cD_{m', n'}$.
 Then $S$ is reachable in $\cD_{m,n}$.
\end{lemma}

\begin{proof}
 Recall that we can
 assume $m \ge 2$ and $n \ge 2$.
 We may assume that there is neither an empty row nor an empty column in $S$; otherwise $S$ is reachable by Lemma~\ref{lem:induction}.
 It is sufficient to consider the case involving a column, 
 since the case involving a row follows by symmetric arguments.
 Let $(p,q)$ be the only element in column $q$.
 If there are more elements in row $p$, 
 then column $q$ is contained in another column and by Lemma~\ref{lem:induction}, the set $S$ is reachable.

 Let $S'$ be the subset of $Q_{m-1} \times Q_{n-1}$ 
 obtained by removing row $p$ and column~$q$, 
 and renumbering the states to $Q_{m-1} \times Q_{n-1}$ 
 in the way such that $i \in Q_m$ becomes $i-1$ if $i>p$ 
 and otherwise remains the same, 
 and $j \in Q_n$ becomes $j-1$ if $j>q$ 
 and otherwise remains the same.
 We have that $S'$ is a valid subset,
 and by the inductive assumption 
 it is reachable in $\cD_{m-1,n-1}$ 
 by some word $u'$; 
 let $u$ be the word corresponding to $u'$ 
 in the original numbering of the states.
 We consider four cases.
 
 Case $p \ne 1$ and $q \ne 1$:
 State $\{(1,1),(p,q)\}$ is reachable 
 in   $\cD_{m,n}$
  by word $a^2$, 
 where $a\colon (1,p); (1,q)$.
 Then $S$ is reachable by $a^2 u$.

 Case $p = 1$ and $q \ne 1$:
 State $\{(2,1),(1,q)\}$ is reachable 
 in $\cD_{m,n}$ by word $a^2$, 
 where $a\colon (1,2); (1,q)$.
 Then state $(2,1)$ corresponds to state $(1,1)$ 
 after the renumbering, and $S$ is reachable by $a^2 u$.

 Case $p \ne 1$ and $q = 1$:
 This is symmetrical to the previous case.

 Case $p = 1$ and $q = 1$:
 State $\{(1,1),(2,2)\}$ is reachable 
 in $\cD_{m,n}$ by word $a^2$, 
 where $a\colon (1,2); (1,2)$.
 Then state $(2,2)$ corresponds to state $(1,1)$ 
 after the renumbering, and $S$ is reachable by $a^2 u$.
\qed
\end{proof}

\begin{theorem}\label{thm:generalization}
 If for some $h$ every valid subset can be reached
 in $\cD_{h,\binom{h}{\lfloor h/2 \rfloor}}$
 then for  every $m \le h$ and every~$n$,
 every valid subset 
 can be reached in $\cD_{m,n}$.
\end{theorem}

\begin{proof}
 This follows by induction on $m$, $n$, and $|S|$.
 
 For $m=1$ this follows by induction on $n$:  
 if $n=1$ then $\cD_{1,1}$ 
 consists of a single valid subset $\{(1,1)\}$,
 and if $n>1$, then we apply   Lemma~\ref{lem:induction}. 
 For $m \le h$ and $n \le \binom{h}{\lfloor h/2 \rfloor}$ 
 this holds by assumption,
 since $\cN_{m,n}$ is a sub-NFA 
 of $\cN_{h,\binom{h}{\lfloor h/2 \rfloor}}$.
 If $|S|=1$, then    $\{(1,1)\}$ is the only valid subset, 
  and it is reachable since it is the initial subset of $\cD_{m,n}$. 

 Let $S$ be a valid subset of $Q_{m} \times Q_{n}$, 
 where $m \le h$ and $n > \binom{h}{\lfloor h/2 \rfloor}$, 
 and assume that every valid subset $S'$ of $Q_{m'} \times Q_{n'}$ 
 is reachable if $m' < m$, or $n' < n$, or $|S'| < |S|$.
 By Sperner's theorem \cite{Spe28}, 
 the maximal number of subsets of an $m$-element set 
 such that none of them contains any other subset is $\binom{m}{\lfloor m/2 \rfloor}$.
 This is not larger than $\binom{h}{\lfloor h/2 \rfloor}$; 
 hence, there exist some columns $j,j'$ with $j\neq j'$
 such that the $j$-th column is contained in $j'$-th column.
 By Lemma~\ref{lem:induction}, the subset $S$ is reachable.
\qed
\end{proof}

\begin{corollary}\label{cor:reachability_m=4}
 Let $1\le m\le4$ and $n\ge1$.
  Then 
 every valid subset can be reached in $\cD_{m,n}$.
\end{corollary}
\begin{proof}
 Since we have verified the reachability 
 of all valid subsets for $m=4$ and $n=6=\binom{4}{2}$, Theorem~\ref{thm:generalization} applies with $h=4$.
\qed
\end{proof}

To strengthen this result 
and show reachability for $m \le 5$, 
we need to introduce another concept with permutations.
Let $\varphi$ be any permutation of $m$ rows. 
We split subsets of $Q_m$ (subsets of rows) 
into equivalence classes under $\varphi$.
For $U \subseteq Q_m$,  
$[U]_\varphi = 
\{V \subseteq Q_m \mid V= \varphi^i(U)\text{ for some }i \ge 0\}$ denotes the equivalence class of $U$.
See Tables~\ref{tab:permutational_example_u09}, \ref{tab:permutational_example_u1}, \ref{tab:permutational_example_u4} for examples of subsets whose columns $U$ are partitioned into equivalence classes under some $\varphi$.

For a subset $S$ of $Q_m \times Q_n$, 
by $\col(S,i)$ we denote the subset of $Q_m$ 
contained in the $i$-th column.
Then $\cols(S) = \bigcup_{1 \le i \le n} \col(S,i)$ 
is the set of the subsets in the columns of $S$.

The following lemma assures reachability 
(under an inductive assumption) 
of a special kind of subsets 
whose columns form only full and empty 
equivalence classes under some permutation $\varphi$.

\begin{lemma}\label{lem:permutation}
 Let $\varphi$ be a permutation of $m$ rows.
 Let $S$ be a valid subset of $Q_m~\times~Q_n$ 
 such that $[U]_\varphi \subseteq \cols(S)$ 
 for every $U\in \cols(S)$, 
 and there is a column $V\in \cols(S)$ 
 such that $|[V]_\varphi|\ge 2$.
 Assume that every
 valid subset $S'$ of $Q_{m'} \times Q_{n'}$, 
 where $m' < m$, or $n' < n$, or $|S'| < |S|$, 
 is reachable in $\cD_{m',n'}$.
 Then $S$ is reachable in $\cD_{m,n}$.
\end{lemma}

\begin{proof}
 We can assume that no two columns contain the same subset of rows, 
 no column is empty, 
 and the first row contains at least two elements; 
 otherwise $S$ is reachable by Lemma~\ref{lem:induction} 
 or by Lemma~\ref{lem:single_element}.

 Let $S_j=\col(S,j)$ be the $j$-th column 
 of  a valid subset $S$. 
 Thus we have 
 $S=\{(i,j)\mid 1\le j \le n \text{ and } i\in S_j \}.$ 
 Since $|[V]_\varphi|\ge 2$, 
 we can always choose $V$ so that $\varphi^{-1}(V)$ 
 is in a $k$-th column $S_k$ with $k\neq1$. 
 Let $S'$ 
 be the set obtained from $S$ 
 by omitting the states in the $k$-th column
 and by taking the pre-image  of $S_j$ under $\varphi$
 in any other column, that is,
 $$
    S'=\{(i,j)\mid 1\le j \le n,
    j\neq k, \text{ and } i\in \varphi^{-1}(S_j) \}. 
 $$
 Since $k\neq 1$ and the first row of $S$
 contains at least two elements,
 the set $S'$ is valid.
 Since $V$ is non-empty, we have $|S'| < |S|$.
 Let $\psi$ be a permutation
 that maps a column $j$ 
 to the column containing $\varphi^{-1}(S_j)$,
 that is, we have
 $S_{\psi(j)}= \varphi^{-1}(S_j)$.
 Let $t$ be the transformation given by $a_{\varphi,\psi}$.
 Let us show that $S' t  = S$.

 Let $(i,j)\in S'$. 
 Then $i\in \varphi^{-1}(S_j)$, so $\varphi(i)\in S_j$,
 and we have 
 $(i,j)t = \{(\varphi(i),j), (i, \psi(j))\}\subseteq S$.
 Hence $S' t \subseteq S$.

 Now let $(i,j)\in S$. 
 First let  $j\neq k$. 
 Then $i\in S_j$, so $\varphi^{-1}(i)\in \varphi^{-1}(S_j)$. 
 Therefore $(\varphi^{-1}(i),j)\in~S'$.
 Since $(i,j)\in (\varphi^{-1}(i),j) t $,
 we have $(i,j)\in S' t$.
 Now let $j=k$. Then $i\in\varphi^{-1}(V) $
 and $S_{\psi^{-1}(k)} = V$. Thus $(i, \psi^{-1}(k))\in S'$,
 and we have $(i,k)\in (i, \psi^{-1}(k)) t $.
 Hence $S\subseteq S't$.
 Our proof is complete.
\qed
\end{proof}

\newcommand{\X}{\ \circ\ }
\begin{table}[ht]
\caption{A subset and the equivalence classes of columns under $\varphi = [2,3,1,4,5]$.}\label{tab:permutational_example_u09}
\begin{center}
$\begin{array}{|l||c|c|c|c|c|c|c|c|c|}\hline
         & 1& 2& 3& 4& 5& 6& 7& 8&9 \\\hline\hline
1        &\X&  &\X&\X&  &  &\X&  &  \\\hline
2        &\X&\X&  &  &\X&  &  &\X&  \\\hline
3        &  &\X&\X&  &  &\X&  &  &\X\\\hline
4        &  &  &  &\X&\X&\X&  &  &  \\\hline
5        &  &  &  &  &  &  &\X&\X&\X\\\hline\hline
\text{eq}& A& A& A& B& B& B& C& C& C\\\hline
\end{array}$
\end{center}
\end{table}
\begin{table}[ht]
\caption{A subset and the equivalence classes of columns under $\varphi = [1,2,3,5,4]$.}\label{tab:permutational_example_u1}
\begin{center}
$\begin{array}{|l||c|c|c|c|c|c|c|c|}\hline
         & 1& 2& 3& 4& 5& 6& 7& 8\\\hline\hline
1        &\X&\X&  &  &\X&  &  &  \\\hline
2        &\X&  &\X&  &  &\X&  &  \\\hline
3        &\X&  &  &\X&  &  &\X&  \\\hline
4        &  &\X&\X&\X&  &  &  &\X\\\hline
5        &  &  &  &  &\X&\X&\X&\X\\\hline\hline
\text{eq}& A& B& C& D& B& C& D& E\\\hline
\end{array}$
\end{center}
\end{table}
\begin{table}[ht]
\caption{A subset and the equivalence classes of columns under $\varphi = [2,3,4,1,5]$.}\label{tab:permutational_example_u4}
\begin{center}
$\begin{array}{|l||c|c|c|c|c|c|c|c|}\hline
         & 1& 2& 3& 4& 5& 6& 7& 8\\\hline\hline
1        &  &\X&\X&\X&\X&  &  &  \\\hline
2        &\X&  &\X&\X&  &\X&  &  \\\hline
3        &\X&\X&  &\X&  &  &\X&  \\\hline
4        &\X&\X&\X&  &  &  &  &\X\\\hline
5        &  &  &  &  &\X&\X&\X&\X\\\hline\hline
\text{eq}& A& A& A& A& B& B& B& B\\\hline
\end{array}$
\end{center}
\end{table}

\begin{corollary}\label{cor:m=5}
 Let $1\le m \le 5$ and $n\ge1$.
  Then every valid subset can be reached in $\cD_{m,n}$.
\end{corollary}
\begin{proof}
The proof follows by analysis of valid subsets $S \subseteq Q_5 \times Q_n$, with the aid of Corollary~\ref{cor:reachability_m=4}, Lemma~\ref{lem:induction}, Lemma~\ref{lem:permutation}, and the results from Table~\ref{tab:reachability}.

 Suppose that there is a valid subset $S \subseteq Q_5 \times Q_n$ that is not reachable; 
 let $S$ be chosen so that $n$ 
 is the smallest number and $S$ 
 is a smallest non-reachable subset of $Q_5 \times Q_n$.

 By Corollary~\ref{cor:reachability_m=4} and the choice of $n$, 
 every valid subset $S' \subset Q_{m'} \times Q_{n'}$, 
 where $m' < 5$, or $n' < n$, or $|S'| < |S|$, is reachable. 
 Hence, $S$ has no column containing another column; 
 otherwise, we can apply Lemma~\ref{lem:induction}.
 Since we have verified the reachability of all valid subsets 
 for $m=5$ and $n\le 7$ (Table~\ref{tab:reachability}), 
 we must have $n \ge 8$ and so $S$ has at least 8 distinct columns.
 Obviously there is neither an empty nor a full column. 
 If there is a column $U$ with $|U|=1$ or $|U|=4$, 
 then by Sperner's theorem if $n > \binom{4}{2} = 6$, 
 then $S$ has a column containing another column; 
 hence $S$ can have only columns $U$ with $|U|=3$ or $|U|=2$.

 Let $C_3$ be the number of 3-element columns ($|U|=3$), 
 and $C_2$ be the number of 2-element columns ($|U|=2$).
 We are searching for possible subsets $S$ 
 that do not have a column containing another column,  
 and with $C_3+C_2 \ge 8$.  
 We consider the following six cases.

\smallskip (1)
 Let $C_3=0$.
 If $C_2 = 10$, which implies 
 that $S$ contains all possible 2-element subsets, 
 then under $\varphi = [2,3,4,5,1]$ 
 we have two full and non-trivial equivalence classes. 
 Hence $S$ is reachable from a smaller subset
 by Lemma~\ref{lem:permutation}. 
 If $C_2 = 9$, then without loss of generality 
 let the missing 2-element subset be $\{4,5\}$; 
 see Table~\ref{tab:permutational_example_u09}. 
 Under $\varphi = [2,3,1,4,5]$ 
 we have three full and non-trivial equivalence classes, 
 and  $S$ is reachable by Lemma~\ref{lem:permutation}. 
 Finally, if $C_2 = 8$, 
 then we have two subcases. 
 If the two missing 2-element subsets have a common element, 
 then without loss of generality 
 let them be $\{2,3\}$ and $\{4,5\}$. 
 Under $\varphi = [1,4,5,2,3]$ 
 we have four full and non-trivial equivalence classes,
 and $S$ is reachable by Lemma~\ref{lem:permutation}.
 If they have a common element, 
 then without loss of generality 
 let them be $\{3,4\}$ and $\{4,5\}$. 
 Under $\varphi = [1,2,5,4,3]$ 
 we have six full equivalence classes 
 and two of them are non-trivial.
 Thus $S$ is reachable by Lemma~\ref{lem:permutation}.

\smallskip (2)
 Let $C_3=1$.  The only possible subset, 
 up to permutation of columns and rows, 
 is shown in Table~\ref{tab:permutational_example_u1}. 
 It has all columns with two elements 
 that are not contained in the 3-element column. 
 By Lemma~\ref{lem:permutation} with $\varphi = [1,2,3,5,4]$, it is reachable.

\smallskip (3)
 Let $C_3=2$. 
 A simple analysis reveals 
 that if the 3-element columns have only one common element, 
 then $C_2$ is at most 4. 
 If they have two common elements, then $C_2$ is at most 5.
 Thus in this case, we have $C_2+C_3\le7$.
 
\smallskip (4)
 Let $C_3=3$. Here $C_2$ is at most 4.
 
 \smallskip (5)
 Let $C_3=4$. The only possible subset, 
 up to permutation of columns and rows,
 is shown in Table~\ref{tab:permutational_example_u4}. 
 By Lemma~\ref{lem:permutation} with $\varphi = [2,3,4,1,5]$, 
 it is reachable.

\smallskip (6)
 Let $C_3\ge 5$. 
 These cases are symmetrical to those with $C_3 \le 3$; 
 it is sufficient to consider the complement of $S$.

 Since these cover all the possibilities for set $S$, 
 this set is reachable.
\qed
\end{proof}

\subsection{Proof of Distinguishability}\label{sec:distinguish}
\label{se:distinguishability}

The aim of this section 
is to show that there are regular languages
 defined over a three-letter alphabet 
such that the subset automaton of the NFA for their shuffle
does not have equivalent states.

To this aim let $\cA = (Q,\Sigma,\delta,s,F)$ be an NFA.
We say that a state $q$ in $Q$ is \emph{uniquely distinguishable}
if there is a word $w$ in $\Sigma^*$
which is accepted by $\cA$ from and only from the state $q$,
that is, if there is a word $w$ such that
$\delta(p,w)\in F$ if~and only if $p=q$.
First, let us prove the following two observations.

\begin{proposition}
\label{prop:disting}
 If each state of an NFA $\cA$ is uniquely distinguishable,
 then the subset automaton of $\cA$
 does not have equivalent states.
\end{proposition}

\begin{proof}
 Let $S$ and $T$ be two distinct subsets  in $2^Q$.
 Then, without loss of generality,
 there is a state $q$ in $Q$ with $q\in S \setminus T$.
 Since $q$ is uniquely distinguishable,
 there is a word $w$ which is accepted by $\cA$ 
 from and only from $q$. 
 Therefore,  the subset automaton of $\cA$
 accepts $w$ from $S$ and it rejects $w$ from $T$.
 Hence $w$ distinguishes   $S$ and $T$. 
\qed 
\end{proof}

\begin{proposition}
\label{prop:unique}
 Let a state $q$ of an NFA $\cA = (Q,\Sigma,\delta,s,F)$
 be uniquely distinguishable.
 Assume that there is a symbol $a$ in $\Sigma$
 and \emph{exactly one} state $p$ in~$Q$  that goes to $q$ on $a$,
 that is, $(p,a,q)$ is a unique in-transition on $a$ going to $q$.
 Then the state $p$ is uniquely distinguishable as well.
\end{proposition}

\begin{proof}
 Let $w$ be a word which is accepted by $\cA$
 from and  only from $q$.
 The word $a w$ is accepted from $p$
 since $q \in \delta(p,a)$ 
 and $w$ is accepted from $q$.
 Let $r\neq p$.
 Then $q\notin\delta(r,a)$
 since $(p,a,q)$ is a unique in-transition on $a$ going to $q$.
 It follows that the word $w$ is not accepted
 from any state in $\delta(r,a)$.
 Thus $\cA$ rejects $a w$ from $r$,
 so $p$ is uniquely distinguishable.
\qed 
\end{proof}

Now we can prove the following result.

\begin{theorem}
 Let $m,n\ge2$.
 There exist ternary languages $K$ and $L$
 with $\kappa(K)=m$ and $\kappa(L)=n$
 such that the subset automaton of the NFA
 accepting
 $K \shu L$
 does not have equivalent states.
\end{theorem}

\begin{proof}
 Let $m$ and $n$  be arbitrary but fixed integers
 with $m,n\ge2$.
 Let $K$ be accepted by the  DFA 
 $\cK =(\{1,2,\ldots,m\},\{a,b,c\},\delta_K,1, \{m\})$,
 where for each $i$ in $\{1,2,\ldots,m\}$,
 
 \medskip
 $\delta_K(i,a)=i+1$ if $i\le m-1$ and $\delta_K(m,a)=1$;
 
 $\delta_K(i,b)=1$;
 
 $\delta_K(1,c)=2$ and $\delta_K(i,c)=1$ if $i\ge2$. 
 
 \medskip\noindent
 Let $L$ be accepted by the  DFA 
 $\cL=(\{1,2,\ldots,n\},\{a,b,c\},\delta_L,1, \{n\})$,
 where for each $j$ in $\{1,2,\ldots,n\}$,
 
 \medskip
 $\delta_L(j,a)=1$;
 
 $\delta_L(j,b)=j+1$ if $j\le n-1$ and $\delta_L(n,b)=1$;
 
 $\delta_L(j,c)=n$. 
 
 \medskip\noindent
 The DFAs $ \cK $ and $\cL$ are shown in Fig.~\ref{fig:dfaKL}.

 \begin{figure}[h]
 \unitlength 8pt
 \begin{center}\begin{picture}(40,14)(0,2)
\gasset{Nh=2.4,Nw=2.4,Nmr=1.2,ELdist=0.4,loopdiam=1.4}

\gasset{iangle=120}
\node(1')(2,14){$1$}\imark(1')
\node(2')(8,14){$2$}
\node(3')(13,14){$3$}
\gasset{Nframe=n}
\node(4')(18,14){$...$}
\gasset{Nframe=y,Nw=3.4}
\node(5')(23.5,14){$m-1$}
\gasset{Nw=2.4}
\node(6')(29,14){$m$}\rmark(6')
\drawedge[ELdist=.2](1',2'){$a,c$}
\drawedge[ELdist=.4](2',3'){$a$}
\drawedge(3',4'){$a$}
\drawedge(4',5'){$a$}
\drawedge(5',6'){$a$}

\drawedge[ELside=l,curvedepth=0.7,ELpos=40,ELdist=.1](2',1'){$b,c$}
\drawbpedge[ELside=l,ELpos=20,ELdist=.1](3',230,4,1',320,4){$b,c$}
\drawbpedge[ELside=l,ELpos=30,ELdist=.1](5',220,5,1',300,5){$b,c$}
\drawbpedge[ELside=l,ELpos=20,ELdist=.1](6',230,6,1',280,6){$a,b,c$}

\gasset{loopangle=180}
\drawloop(1'){$b$}

\node(1)(11,4){$1$}\imark(1)
\node(2)(17,4){$2$}
\node(3)(22,4){$3$}
\gasset{Nframe=n}
\node(4)(27,4){$...$}
\gasset{Nframe=y,Nw=3.5}
\node(5)(32,4){$n-1$}
\gasset{Nw=2.4}
\node(6)(38,4){$n$}\rmark(6)
\drawedge[ELdist=.4](1,2){$b$}
\drawedge(2,3){$b$}
\drawedge(3,4){$b$}
\drawedge(4,5){$b$}
\drawedge[ELdist=.1](5,6){$b,c$}

\drawedge[ELside=l,curvedepth=0.7,ELpos=40,ELdist=.3](2,1){$a$}
\drawbpedge[ELside=l,ELpos=20,ELdist=.2](3,230,4,1,320,4){$a$}
\drawbpedge[ELside=l,ELpos=30,ELdist=.2](5,220,5,1,300,5){$a$}
\drawbpedge[ELside=l,ELpos=20,ELdist=.1](6,230,6,1,280,6){$a,b$}

\drawbpedge[ELside=l,ELpos=30,ELdist=.2](3,50,4,6,150,5){$c$}
\drawbpedge[ELside=l,ELpos=25,ELdist=.2](2,50,5,6,130,6){$c$}
\drawbpedge[ELside=l,ELpos=20,ELdist=.3](1,50,6,6,110,7){$c$}

\gasset{loopangle=180}
\drawloop(1){$b$}
\gasset{loopangle=0}
\drawloop(6){$c$}

\end{picture}\end{center}
 \caption{The   DFAs $ \cK  $ and $\cL $.}
 \label{fig:dfaKL}
 \end{figure}

 Construct the NFA  $\cN$ for  $K\shu L$ 
 as described in Section~\ref{sec:bound} 
 on page~\pageref{nfaN}.
 The transitions on $a,b,c$ in   $\cN $
 for $m=4$ and $n=5$ are shown in Fig.~\ref{fig:NFAabc}.
 Notice that  each state $(i,j)$ 
 with $2\le i \le m$ and $2 \le j \le n$
 has a unique in-transition on symbol $a$
 and this transition goes from   state $(i-1,j)$;
 see the dashed transitions in Fig.~\ref{fig:NFAabc} (top-left).
 Next, 
 each state $(m,j)$ with $2\le j\le n$
 has a unique in-transition on $b$ which goes from $(m,j-1)$,
 and each state $(i,2)$ with $2\le i\le m$
 has a unique in-transition on $b$ going from $(i,1)$;
 see the dashed transitions in Fig.~\ref{fig:NFAabc} (top-right).
 Finally, the state $(2,1)$ has a unique in-transition on $c$
 going from $(1,1)$; see the dashed transition in Fig.~\ref{fig:NFAabc} (bottom).
 
 The empty word is accepted by   $\cN $ from and  only from  the state   $(m,n)$ 
 since this is a unique accepting state of $\cN$.
 Thus $(m,n)$ is uniquely distinguishable.
 Next, consider the subgraph of unique in-transitions in $\cN$.
 Fig.~\ref{fig:nfa_unique} shows  this subgraph in the case of $m=4$ and $n=5$.
 Notice that from each state of $\cN$,
 the state $(m,n)$ is reachable in this subgraph.
 By Proposition~\ref{prop:unique}, used repeatedly,
 we get that  each state of $\cN$
 is uniquely distinguishable. 
 Hence by Proposition~\ref{prop:disting},
 the subset automaton of $\cN $
 does not have equivalent states.
\qed
\end{proof}
  
 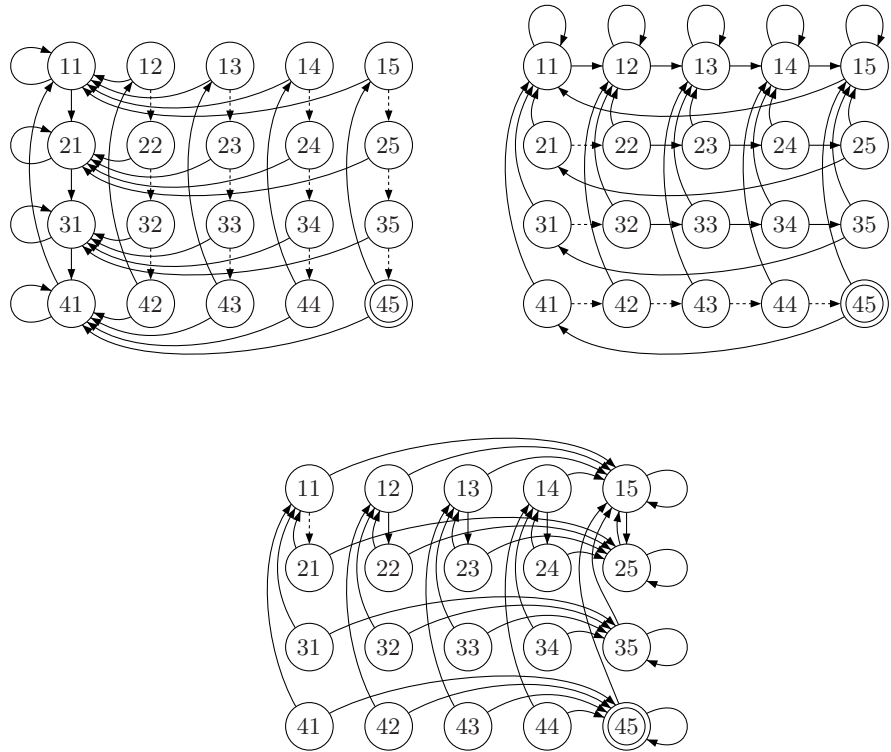
\begin{figure}[h!]
 \vskip-20pt
 \unitlength 10pt
 \begin{center}\begin{picture}(34,28)(-1,14)
\gasset{Nh=1.8,Nw=1.8,Nmr=0.9,linewidth=0.04,loopdiam=1.4}

\node(11')(2,40){$11$}
\node(12')(5,40){$12$}
\node(13')(8,40){$13$}
\node(14')(11,40){$14$}
\node(15')(14,40){$15$}
\node(11'')(20,40){$11$}
\node(12'')(23,40){$12$}
\node(13'')(26,40){$13$}
\node(14'')(29,40){$14$}
\node(15'')(32,40){$15$}

\gasset{loopangle=180}
\drawloop(11'){}

\drawedge(11'',12''){}
\drawedge(12'',13''){}
\drawedge(13'',14''){}
\drawedge(14'',15''){}

\gasset{loopangle=90}
\drawloop(11''){}
\drawloop(12''){}
\drawloop(13''){}
\drawloop(14''){}
\drawloop(15''){}

\drawbpedge(12',220,1.6,11',340,1.6){}
\drawbpedge(13',220,2.5,11',320,2.5){}
\drawbpedge(14',220,2.8,11',300,2.8){}
\drawbpedge(15',220,3.1,11',280,3.1){}
\drawbpedge(15'',220,3.1,11'',280,3.1){}

\node(21')(2,37){$21$}
\node(22')(5,37){$22$}
\node(23')(8,37){$23$}
\node(24')(11,37){$24$}
\node(25')(14,37){$25$}
\node(21'')(20,37){$21$}
\node(22'')(23,37){$22$}
\node(23'')(26,37){$23$}
\node(24'')(29,37){$24$}
\node(25'')(32,37){$25$}

\drawbpedge(22',220,1.6,21',340,1.6){}
\drawbpedge(23',220,2.5,21',320,2.5){}
\drawbpedge(24',220,2.8,21',300,2.8){}
\drawbpedge(25',220,3.1,21',280,3.1){}
\drawbpedge(25'',220,3.1,21'',280,3.1){}

\drawedge(11',21'){}
\gasset{linewidth=0.04, dash={0.12}0 }
\drawedge(12',22'){}
\drawedge(13',23'){}
\drawedge(14',24'){}
\drawedge(15',25'){}
\drawedge(21'',22''){}

\gasset{linewidth=0.04,dash={}{0}}
\gasset{loopangle=180}
\drawloop(21'){}
\drawedge(22'',23''){}
\drawedge(23'',24''){}
\drawedge(24'',25''){}

\drawbpedge(21'',130,1.6,11'',250,1.6){}
\drawbpedge(22'',130,1.6,12'',250,1.6){}
\drawbpedge(23'',130,1.6,13'',250,1.6){}
\drawbpedge(24'',130,1.6,14'',250,1.6){}
\drawbpedge(25'',130,1.6,15'',250,1.6){}

\node(31')(2,34){$31$}
\node(32')(5,34){$32$}
\node(33')(8,34){$33$}
\node(34')(11,34){$34$}
\node(35')(14,34){$35$}
\node(31'')(20,34){$31$}
\node(32'')(23,34){$32$}
\node(33'')(26,34){$33$}
\node(34'')(29,34){$34$}
\node(35'')(32,34){$35$}

\drawbpedge(32',220,1.6,31',340,1.6){}
\drawbpedge(33',220,2.5,31',320,2.5){}
\drawbpedge(34',220,2.8,31',300,2.8){}
\drawbpedge(35',220,3.1,31',280,3.1){}
\drawbpedge(35'',220,3.1,31'',280,3.1){}

\drawedge(21',31'){}
\gasset{linewidth=0.04, dash={0.12}0 }
\drawedge(22',32'){}
\drawedge(23',33'){}
\drawedge(24',34'){}
\drawedge(25',35'){}
\drawedge(31'',32''){}
\gasset{linewidth=0.04,dash={}{0}}

\gasset{loopangle=180}
\drawloop(31'){}
\drawedge(32'',33''){}
\drawedge(33'',34''){}
\drawedge(34'',35''){}

\drawbpedge(31'',130,2.5,11'',230,2.5){}
\drawbpedge(32'',130,2.5,12'',230,2.5){}
\drawbpedge(33'',130,2.5,13'',230,2.5){}
\drawbpedge(34'',130,2.5,14'',230,2.5){}
\drawbpedge(35'',130,2.5,15'',230,2.5){}

\node(41')(2,31){$41$}
\node(42')(5,31){$42$}
\node(43')(8,31){$43$}
\node(44')(11,31){$44$}
\node(45')(14,31){$45$}\rmark(45')
\node(41'')(20,31){$41$}
\node(42'')(23,31){$42$}
\node(43'')(26,31){$43$}
\node(44'')(29,31){$44$}
\node(45'')(32,31){$45$}\rmark(45'')

\drawbpedge(42',220,1.6,41',340,1.6){}
\drawbpedge(43',220,2.5,41',320,2.5){}
\drawbpedge(44',220,2.8,41',300,2.8){}
\drawbpedge(45',220,3.1,41',280,3.1){}
\drawbpedge(45'',220,3.1,41'',280,3.1){}

\drawedge(31',41'){}
\gasset{linewidth=0.04, dash={0.12}0 }
\drawedge(32',42'){}
\drawedge(33',43'){}
\drawedge(34',44'){}
\drawedge(35',45'){}
\drawedge(41'',42''){}
\drawedge(42'',43''){}
\drawedge(43'',44''){}
\drawedge(44'',45''){}
\gasset{linewidth=0.04,dash={}{0}}

\gasset{loopangle=180}
\drawloop(41'){}
\drawbpedge(41',130,2.8,11',210,2.8){}
\drawbpedge(42',130,2.8,12',210,2.8){}
\drawbpedge(43',130,2.8,13',210,2.8){}
\drawbpedge(44',130,2.8,14',210,2.8){}
\drawbpedge(45',130,2.8,15',210,2.8){}
\drawbpedge(41'',130,2.8,11'',210,2.8){}
\drawbpedge(42'',130,2.8,12'',210,2.8){}
\drawbpedge(43'',130,2.8,13'',210,2.8){}
\drawbpedge(44'',130,2.8,14'',210,2.8){}
\drawbpedge(45'',130,2.8,15'',210,2.8){}


\node(11)(11,24){$11$}
\node(12)(14,24){$12$}
\node(13)(17,24){$13$}
\node(14)(20,24){$14$}
\node(15)(23,24){$15$}

\node(21)(11,21){$21$}
\node(22)(14,21){$22$}
\node(23)(17,21){$23$}
\node(24)(20,21){$24$}
\node(25)(23,21){$25$}

\node(31)(11,18){$31$}
\node(32)(14,18){$32$}
\node(33)(17,18){$33$}
\node(34)(20,18){$34$}
\node(35)(23,18){$35$}

\node(41)(11,15){$41$}
\node(42)(14,15){$42$}
\node(43)(17,15){$43$}
\node(44)(20,15){$44$}
\node(45)(23,15){$45$}\rmark(45)

\gasset{linewidth=0.04, dash={0.12}0 }
\drawedge(11,21){}
\gasset{linewidth=0.04,dash={}{0} }
\drawedge(12,22){}
\drawedge(13,23){}
\drawedge(14,24){}
\drawedge(15,25){}

\drawbpedge(14,40,1.6,15,160,1.6){}
\drawbpedge(13,40,2.5,15,140,2.5){}
\drawbpedge(12,40,2.8,15,120,2.8){}
\drawbpedge(11,40,3.1,15,100,3.1){}
\gasset{loopangle=0}
\drawloop(15){}

\drawbpedge(24,40,1.6,25,160,1.6){}
\drawbpedge(23,40,2.5,25,140,2.5){}
\drawbpedge(22,40,2.8,25,120,2.8){}
\drawbpedge(21,40,3.1,25,100,3.1){}
\gasset{loopangle=0}
\drawloop(25){}

\drawbpedge(34,40,1.6,35,160,1.6){}
\drawbpedge(33,40,2.5,35,140,2.5){}
\drawbpedge(32,40,2.8,35,120,2.8){}
\drawbpedge(31,40,3.1,35,100,3.1){}
\gasset{loopangle=0}
\drawloop(35){}

\drawbpedge(44,40,1.6,45,160,1.6){}
\drawbpedge(43,40,2.5,45,140,2.5){}
\drawbpedge(42,40,2.8,45,120,2.8){}
\drawbpedge(41,40,3.1,45,100,3.1){}
\gasset{loopangle=0}
\drawloop(45){}

\drawbpedge(21,130,1.6,11,250,1.6){}
\drawbpedge(31,130,2.5,11,230,2.5){}
\drawbpedge(41,130,2.8,11,210,2.8){}

\drawbpedge(22,130,1.6,12,250,1.6){}
\drawbpedge(32,130,2.5,12,230,2.5){}
\drawbpedge(42,130,2.8,12,210,2.8){}

\drawbpedge(23,130,1.6,13,250,1.6){}
\drawbpedge(33,130,2.5,13,230,2.5){}
\drawbpedge(43,130,2.8,13,210,2.8){}

\drawbpedge(24,130,1.6,14,250,1.6){}
\drawbpedge(34,130,2.5,14,230,2.5){}
\drawbpedge(44,130,2.8,14,210,2.8){}

\drawbpedge(25,105,1.6,15,250,1.6){}
\drawbpedge(35,105,2.5,15,230,4){}
\drawbpedge(45,105,2.3,15,210,4.3){}

\end{picture}\end{center}
 \caption{NFA $\cN $ for $m=4$ and $n=5$; 
 	the transitions on $a$ (top-left), $b$ (top-right), and $c$ (bottom).}
 \label{fig:NFAabc}
 \end{figure}
  
 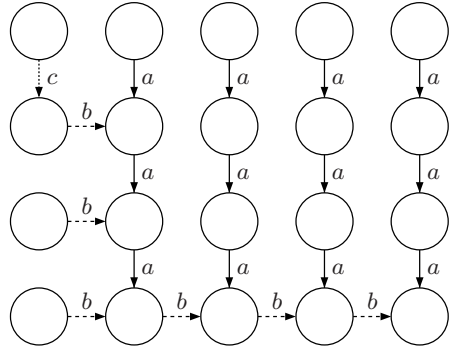
\begin{figure}[t]
 \unitlength 12pt
 \begin{center}\begin{picture}(20,13)(7,14)
 \gasset{Nh=1.8,Nw=1.8,Nmr=0.9,linewidth=0.04,loopdiam=1.4}
 
 \node(11)(11,24){} 
 \node(12)(14,24){} 
 \node(13)(17,24){} 
 \node(14)(20,24){} 
 \node(15)(23,24){} 
 
 \node(21)(11,21){} 
 \node(22)(14,21){} 
 \node(23)(17,21){} 
 \node(24)(20,21){} 
 \node(25)(23,21){} 
 
 \node(31)(11,18){} 
 \node(32)(14,18){} 
 \node(33)(17,18){} 
 \node(34)(20,18){} 
 \node(35)(23,18){} 
 
 \node(41)(11,15){} 
 \node(42)(14,15){} 
 \node(43)(17,15){} 
 \node(44)(20,15){} 
 \node(45)(23,15){} 

 \gasset{dash={0.06}0 }
 \drawedge(11,21){$c$}
 
 \gasset{dash={0.12}0 }
 \drawedge(21,22){$b$}
 \drawedge(31,32){$b$}
 \drawedge(41,42){$b$}
 \drawedge(42,43){$b$}
 \drawedge(43,44){$b$}
 \drawedge(44,45){$b$}
 \gasset{dash={}{0} }
 
 \drawedge(12,22){$a$}
 \drawedge(13,23){$a$}
 \drawedge(14,24){$a$}
 \drawedge(15,25){$a$}
 
 \drawedge(22,32){$a$}
 \drawedge(23,33){$a$}
 \drawedge(24,34){$a$}
 \drawedge(25,35){$a$}
 
 \drawedge(32,42){$a$}
 \drawedge(33,43){$a$}
 \drawedge(34,44){$a$}
 \drawedge(35,45){$a$}

 \end{picture}\end{center}
 \caption{The subgraph of unique in-transitions in NFA $\cN$; $m=4$ and $n=5$. }
 \label{fig:nfa_unique}
 \end{figure}

\section{Conclusions}

We have examined the state complexity of the shuffle operation on two regular languages of state complexities $m$ and $n$, respectively, and found an upper bound for it.
We know that this bound can be reached for any $m$ with $1 \le m \le 5$ and any $n \ge 1$, and also for $m=n=6$.
For the remaining values of $m$ and $n$, however, the problem remains open.
Since there exist two languages $K$ and $L$ for which all pairs of states in the subset automaton of the NFA accepting the shuffle $K\shu L$ are distinguishable, the main difficulty consists of proving that all valid states in the subset automaton can be reached for the witness languages.

\section*{Acknowledgments}

We would like to thank an anonymous referee 
for proposing the notions of a uniquely distinguishable state  
and of a subgraph of unique in-transitions
which allow us to simplify the proof of distinguishability.
We are also   grateful for his comments and suggestions
that helped us  improve the presentation of the paper.

\bibliographystyle{splncs03}

\end{document}